\documentclass[a4paper]{article}
\title{Lemmas of Differential Privacy}
\usepackage[T1]{fontenc}
\usepackage[utf8]{inputenc}

\author{
    Yiyang Huang\\
    {University of New South Wales}\\
    \texttt{sophia.huang1@student.unsw.edu.au}
    \and
    Cl\'ement L. Canonne\\
    {University of Sydney}\\
    \texttt{clement.canonne@sydney.edu.au}
}
\usepackage{styles}
\usepackage{cochineal} %

\begin{document}
    \maketitle
    
    \begin{abstract}
        We aim to collect buried lemmas that are useful for proofs. In particular, we try to provide self-contained proofs for those lemmas and categorise them according to their usage.
    \end{abstract}

    \tableofcontents

    \section{Introduction}
        Differential privacy~\cite{DworkMNS06} ensures that the result of an algorithm is statistically insensitive to changes in its input. The core differential privacy framework has been expanded with additional constraints based on the different stakeholders involved, such as central-, local-, shuffle-, and many others. Nonetheless, due to the sheer amount of work in this field, some useful results are ``buried'' and remain somewhat unknown, and some are repeatedly proved. We hope this survey can serve as a helpful compendium, ideally curating some valuable facts for various differential privacy settings.

    \section{Preliminaries} \label{preliminaries}
        We first recall the notions of central, local, and shuffle privacy. In what follows, two datasets $X, X'\in\mathcal{X}^n$ consisting of $n$ entries are said to be \emph{neighboring} (denoted $X\sim X'$) if they differ in exactly one entry~--~i.e., are at Hamming distance one.
        \subsection{Central Differential Privacy}
            \begin{definition}[(Central) Differential Privacy]
                Fix $\eps>0$ and $\delta\in[0,1]$. An algorithm $M\colon \mathcal{X}^n \to \mathcal{Y}$ satisfies \emph{$(\eps, \delta)$-differential privacy} (DP) if for every pair of neighboring datasets $X, X'$, and every (measurable) subset $S \subseteq \mathcal{Y}$:
                \[
                    \Pr[M(X) \in S] \leq e^\eps \Pr[M(X') \in S] + \delta.
                \]
                We further say that $M$ satisfies \emph{pure} differential privacy ($\eps$-DP) if $\delta = 0$, otherwise it is \emph{approximate} differential privacy.
            \end{definition}

        \subsection{Local Differential Privacy}
            \begin{definition}[Local Differential Privacy]
                An algorithm (local randomizer) $R\colon \mathcal{X} \to \mathcal{Y}$ satisfies \emph{$(\eps, \delta)$-local differential privacy} (LDP) if for any two data records, $x, x'\in\mathcal{X}$, and every subset $S \subseteq \mathcal{Y}$:
                \[
                    \Pr[R(x) \in S] \leq e^\eps \Pr[R(x') \in S] + \delta.
                \]
                We define pure and approximate LDP analogously to the (Central) DP case.
                \end{definition}
            The definition of Local DP (LDP) above is sometimes referred to as \emph{Replacement Local Differential Privacy}, to contrast it with the variant below, \emph{Deletion} LDP,  introduced in \cite{erlingsson_encode_2020}. 
    
            \begin{definition}[Deletion/Removal LDP]
                An algorithm $R\colon \mathcal{X} \to \mathcal{Y}$ is a \emph{deletion} $(\eps, \delta)$-differentially private local randomizer if there exists a random variable (also known as the reference distribution) $R_0$ such that for all $S \subseteq \mathcal{Y}$ and for all $x \in \mathcal{X}$:
                \[
                    e^{-\eps} (\Pr[R_0 \in S] - \delta) \leq \Pr[R(x) \in S] \leq e^\eps \Pr[R_0 \in S] + \delta.
                \]
            \end{definition}
            Unless explicitly specified otherwise, a local randomizer is always (replacement) LDP. One advantage of the notion of deletion LDP is that it allows in some case to obtain tighter results, e.g., by constant factors.
        \subsection{Shuffle Differential Privacy}
            \begin{definition}[Shuffle Model]
                A protocol $P=(R,S,A)$ in the shuffle model consists of three building blocks, and applies to the data from $n$ users, where user $i$ has data $x_i \in \mathcal{X}$. 
                    \begin{enumerate}
                        \item Each user applies the \emph{local randomizer} $R\colon \mathcal{X} \to \mathcal{Y}^\ast$ to their data and reports messages $(y_{i, 1}, \ldots, y_{i, m}) \gets R(x_i)$. {(The number of messages $m$ can itself be randomized.)}
                        \item The tuple of all resulting messages $(y_{i, j})_{i,j}$ is sent to the \emph{shuffler} $S\colon \mathcal{Y}^* \to \mathcal{Y}^*$ that takes in these messages and outputs them in a uniformly random order.
                        \item The shuffled tuple of messages is then passed through some \emph{analyzer} $A\colon \mathcal{Y}^* \to \mathcal{Z}$ to estimate some function $f(x_1, \ldots, x_n)$.
                    \end{enumerate}
            \end{definition}
    
            \begin{definition}[Shuffle Differential Privacy]
                A protocol $P = (R, S, A)$ is $(\eps, \delta)$-shuffle differentially private if, for all $n \in \N$, the shuffled sequence of local randomizer reports, $S \circ R \eqdef S(R(x_1), \ldots, R(x_n))$, is $(\eps, \delta)$-differentially private.
            \end{definition}

            Note that the definition of Shuffle Differential Privacy does not guarantees anything in the presence of malicious users, i.e., users which may deviate from the protocol (not use the randomizer $R$, but something else) in order to jeopardize the privacy of \emph{other} (honest) users. To account for this, the variant of \emph{Robust} Shuffle Differential Privacy was introduced, which ensures privacy even when only a fraction of the users do follow the protocol.
            \begin{definition}[Robust Shuffle Differential Privacy]
                Fix $\gamma \in (0, 1]$. A protocol $P = (R, S, A)$ is \emph{$(\eps, \delta, \gamma)$-robustly shuffle differentially private} if, for all $n \in \N$ and $\gamma' \geq \gamma$, the algorithm $S \circ R$, while only taking $\gamma \cdot n$ users input, is $(\eps, \delta)$-differentially private. i.e., $P$ guarantees $(\eps, \delta)$-shuffle differential privacy when at least a $\gamma$ fraction of users follow the protocol.
            \end{definition}
            While many shuffle private protocols which can be found in the literature do satisfy this robust version ``out-of-the-box'', the two notions are not equivalent and there exist shuffle private protocols which are not robust: see, e.g.,~\cite[Appendix~C]{Cheu21}.
            
    \subsection{Private- and public-coin protocols}
        In both the local and shuffle DP settings, all users run a local randomizer on their private data. This is captured by formally letting the randomizer $R$ be a deterministic mapping $R\colon \mathcal{X}\times\{0,1\}^\ast\to\mathcal{Y}$, where the second input is the randomness, i.e., a string of uniformly random bits, assumed private (``private coins''). Then, user $i$, having data $x_i$ and ``private'' randomness $r_i$, computes $R(x_i,r_i)$, and the $(\eps, \delta)$-DP guarantee must hold as the probability is taken over the random choice of $r_i$.
        
        One can also assume that a source of \emph{public} randomness (``public coins'') is available to all parties (users, analyzer, and outside world alike), acting as a common (non-private) random seed. In this case, the users still have access to their own, private randomness (necessary to ensure differential privacy), but also to an additional input, the public random string (possibly useful to achieve better accuracy or utility). In this case, the ``public-coin'' setting, $R$ can be seen as a deterministic mapping of the form 
            \[
                R\colon \mathcal{X}\times\{0,1\}^\ast\times\{0,1\}^\ast\to\mathcal{Y}
            \]
        and user $i$, having data $x_i$, ``private'' randomness $r_i$, and common (to all users) randomness $r_\text{pub}$, computes $R(x_i, r_i, r_\text{pub})$. Then, the $(\eps, \delta)$-DP guarantee must hold when the probability is taken over the random choice of $r_i$, for every fixed setting of $r_\text{pub}$.

    \section{Central Differential Privacy} \label{dp}
        \begin{lemma}[Pure DP Implies Approximate DP. {\cite[Lemma 5]{acharya_differentially_2017}}]
            Any $(\eps + \delta, 0)$-differentially private algorithm is also $(\eps, \delta)$-differentially private.
        \end{lemma} 
        \begin{proof}
            Suppose $A$ is a $(\eps + \delta, 0)$-DP algorithm, then for any neighboring dataset $X, X'$ and any $S \subseteq \text{range}(A)$, we have:
            \begin{align*}
                \Pr[A(X) \in S]
                &\leq e^{\eps + \delta} \Pr[A(X') \in S]\\
                & = e^\eps \Pr[A(X') \in S] + (e^\delta - 1)e^\eps \Pr[A(X') \in S].
            \end{align*}
            We want to show that $\Pr[A(X) \in S] \leq e^\eps \Pr[A(X') \in S] + \delta$. Since the inequality is trivially true if $e^\eps \Pr[A(X') \in S] + \delta > 1$, we can assume $e^\eps\Pr[A(X') \in S] \leq 1-\delta$. 
            The proof is complete if we show that $(e^\delta - 1)e^\eps \Pr[A(X') \in S] \leq \delta$.  In this case,
            \begin{align*}
                (e^\delta - 1)e^\eps \Pr[A(X') \in S] &\leq (e^\delta - 1)(1-\delta) \leq (e^\delta - 1)e^{-\delta}
                = 1-e^{-\delta} \leq \delta
            \end{align*}
            where the last two inequalities use the fact that $1-x \leq e^{-x}$ for all $x\in\R$.        
        \end{proof}
        The following is a ``folklore'' result, showing that one can convert approximate DP to pure DP for mechanisms with finite output space (and so, in particular, decision algorithms).
        \begin{lemma}[Approximate DP Implies Pure DP for Finite Output Spaces]
            Suppose $A\colon \mathcal{X}\to\mathcal{Y}$ is an $(\eps,\delta)$-DP mechanism, where $\abs{\mathcal{Y}}=k \in \N$. Then, for every $\eta \in (0,1]$, there exists an $\eps'$-DP mechanism $A'\colon \mathcal{X}\to\mathcal{Y}$ such that, for every $x\in\mathcal{X}$, $\totalvardist{A(x)}{A'(x)} \leq \eta$, where $\eps' \eqdef \eps + \ln(1+\frac{\delta k}{\eta}e^{-\eps})$. In particular, for $\mathcal{Y} = \{0,1\}$, $\eps' \leq \eps + \frac{2\delta}{\eta}$.
        \end{lemma}
        \begin{proof}
            Let $A$ be an $(\eps,\delta)$-DP mechanism as in the statement, and define $A'$ as the algorithm which, on input $x\in\mathcal{X}$, outputs $A(x)$ with probability $1-\eta$
            and, otherwise, outputs an element of $\mathcal{Y}$ uniformly at random. We have that, for every $x$, denoting by $\mathbf{u}_{\mathcal{Y}}$ the uniform distribution on $\mathcal{Y}$, 
            \begin{align*}
                \totalvardist{A(x)}{A'(x)} 
                &= \frac{1}{2}\sum_{y\in\mathcal{Y}} \abs{\Pr[A(x)=y] - ((1-\eta)\Pr[A(x)=y]+\frac{\eta}{k})} \\
                &= \eta\cdot  \frac{1}{2}\sum_{y\in\mathcal{Y}} \abs{\Pr[A(x)=y]-\frac{1}{k}} 
                = \eta \cdot \totalvardist{A(x)}{\mathbf{u}_{\mathcal{Y}}} \\
                &\leq \eta\,.
            \end{align*}
            Note that $A'$ is also  $(\eps,\delta)$-DP by postprocessing; and further, for every $x$, we now have $\Pr[A'(x)\in S] \geq \eta\cdot \frac{|S|}{k}$ for every $S\subseteq \mathcal{Y}$. Together, the two imply that, for every $S$ and any two neighboring $x\sim x'$,
            \begin{align*}
                \Pr[A'(x')\in S] &\leq e^\eps \Pr[A'(x)\in S] + \delta \\
                                &= e^\eps \Pr[A'(x)\in S] + \frac{k\delta}{|S|\eta} \cdot \frac{\eta|S|}{k}\\
                                &\leq e^\eps \Pr[A'(x)\in S] + \frac{k\delta}{|S|\eta} \Pr[A'(x)\in S] \\
                                &\leq \mleft(e^\eps + \frac{k\delta}{\eta} \mright) \Pr[A'(x)\in S] 
                                = e^{\eps'}\Pr[A'(x)\in S] 
            \end{align*}
            where $\eps' \eqdef \ln\mleft(e^\eps + \frac{k\delta}{\eta} \mright) = \eps + \ln\mleft(1+\frac{k\delta}{\eta}e^{-\eps}\mright)$.
        \end{proof}

        The following lemma is a useful decomposition proved by Kairouz, Oh, and Viswanath in \cite[Section 6.1]{DBLP:conf/icml/KairouzOV15}. We state the simplified version of \cite[Section 6.1]{DBLP:conf/icml/KairouzOV15} using notation similar to \cite[Lemma~3.4]{feldman_hiding_2021}.
        
        \begin{lemma}[Generalized Version of Decomposing DP Mechanism as Leaky Randomized Response {\cite[Section 6, Definition 3.1]{DBLP:conf/icml/KairouzOV15, DBLP:journals/toc/MurtaghV18}}]
            Let $R\colon \mathcal{X} \to \mathcal{Y}$ be an $(\eps, \delta)$ differentially private mechanism. Then, for all neighboring $x, x' \in \mathcal{X}$, there exists a randomized algorithm $Q\colon \{0, 1, \text{``I am $x$''}, \text{``I am $x'$''}\} \to \mathcal{Y}$ such that:
            \begin{align*}
                R(x) &= \frac{ (1 - \delta)e^{\eps}}{e^{\eps} + 1}Q(0) + \frac{(1 - \delta)}{e^{\eps} + 1}Q(1) + \delta Q(\text{``I am $x$''}), \\
                R(x') &= \frac{(1 - \delta)}{e^{\eps} + 1}Q(0) + \frac{(1 - \delta)e^{\eps}}{e^{\eps} + 1}Q(1) + \delta Q(\text{``I am $x'$''}).
            \end{align*}
        \end{lemma}
        Note that $Q$ above is allowed to depend on $R$ as well as $x,x'$. 
         The proof of this lemma is rather involved, hence we direct the readers to \cite{DBLP:conf/icml/KairouzOV15} for a proof under their hypothesis testing formulation, or else see \cite[Lemma~3.2]{DBLP:journals/toc/MurtaghV18} for an argument of a different flavor.

        It is easy to see that Lemma 3.4 stated in \cite{feldman_hiding_2021} follows as a corollary where $R$ is a $(\eps, 0)$-LDP randomizer with $x, x'$ of size $1$.
        \begin{corollary}[LDP Randomizer as Postprocessing of Binary Randomized Response {\cite[Lemma 3.4]{feldman_hiding_2021}}]
            Let $R\colon \mathcal{X} \to \mathcal{Y}$ be an $\eps$-LDP local randomizer and for all $x, x' \in \mathcal{X}$. Then there exists a randomized algorithm $Q: \{0, 1\} \to \mathcal{Y}$ such that $R(x) = \frac{e^{\eps}}{e^{\eps} + 1}Q(0) + \frac{1}{e^{\eps} + 1}Q(1)$ and $R(x') = \frac{1}{e^{\eps} + 1}Q(0) + \frac{e^{\eps}}{e^{\eps} + 1}Q(1)$.
        \end{corollary}

    \section{Local Differential Privacy} \label{ldp}

        Readers should note that the LDP results we stated in this section are all non-interactive. The interactive case would be more technical, we did not include the most general case here since it's not the aim of this survey.

        \subsection{Replacement and Removal LDP} \label{replace-removal-ldp}

            Here, we present some results involving replacement-deletion notion of local differential privacy.

            \begin{lemma}[Replacement implies Deletion]
                Every $(\eps,\delta)$-replacement LDP randomizer is also an $(\eps,\delta)$-deletion LDP randomizer.
            \end{lemma}
            It is easy to see that a replacement LDP randomizer $R$ is also a deletion $\eps$-differentially private local randomizer by fixing a $x_0$ such that $R_0 = R(x_0)$.
    
            \begin{lemma}[Deletion implies Replacement]
                \label{lemma:ldp:deletion:replacement}
                Every $(\eps,\delta)$-deletion LDP randomizer is also an $(2\eps,(e^\eps+1)\delta)$-replacement LDP randomizer. In particular, every $\eps$-deletion LDP randomizer is also an $2\eps$-replacement LDP randomizer.
            \end{lemma}
            \begin{proof}
                For a reference distribution $R_0$ and $\forall x \in \mathcal{X}, S \subseteq \mathcal{Y}$ since $R\colon \mathcal{X} \to \mathcal{Y}$ is $(\eps, \delta)$-deletion LDP, we have:
                \begin{align*}
                    \Pr[R(x) \in S]
                    &\leq e^{\eps}\Pr[R_0 \in S] + \delta\\
                    &\leq e^{\eps} \big( e^\eps \Pr[R(x') \in S] + \delta \big) + \delta\\
                    &= e^{2\eps}\Pr[R(x') \in S] + \delta(1 + e^\eps).
                \end{align*}
                The second inequality follows from the fact that $R_0$ is the reference distribution for all $x \in \mathcal{X}$, thus it is also true for $x' \in \mathcal{X}$.
            \end{proof}

            \cite{erlingsson_encode_2020} stated, in passing, that ``every $(\eps, \delta)$-deletion LDP randomizer implies $(2\eps, 2\delta)$-replacement LDP''; however, this is not true in general. We provide an counterexample of this statement in the following:

            \emph{Counterexample}: Consider, for $\eps \in [0,1/2]$ and $\delta\in[0,1/5]$ (so that all probabilities below are indeed probabilities), the randomizer $R\colon\{0,1\}\to \{1,2,3\}$ such that
            \[
                R(0) = 
                \begin{cases}
                1 &\text{ w.p. } \frac{e^{\eps}}{3}  \\
                2 &\text{ w.p. } e^{-\eps}\left(\frac{1}{3} -\delta\right) \\
                3 &\text{ w.p. } 1- \left(\frac{e^{\eps}}{3}+ e^{-\eps}\left(\frac{1}{3} -\delta\right)\right) \\
                \end{cases}, \qquad 
                R(1) = 
                \begin{cases}
                1 &\text{ w.p. } \frac{e^{-\eps}}{3}  \\
                2 &\text{ w.p. } \frac{e^{\eps}}{3}+\delta \\
                3 &\text{ w.p. } 1- \left(\frac{e^{-\eps}}{3}+ \frac{e^{\eps}}{3}+\delta\right) \\
                \end{cases}
            \]
            with reference distribution $R_0$ uniform on $\{1,2,3\}$. For instance, for $\eps=1/4$ and $\delta=1/6$, one can check that $R$ as above is indeed an $(\eps, \delta)$-deletion LDP randomizer with reference distribution $R_0$; however, we have
            \[
            \Pr[R(1) = 2] = e^{2\eps} \Pr[R(0) = 2] + \delta(1 + e^\eps) > e^{2\eps} \Pr[R(0) = 2] + 2\delta
            \]
            so $R$ is not $(2\eps,2\delta)$-replacement LDP. \qed
            
            Finally, we conclude with the following lemma, which states that every $(\eps,\delta)$-deletion LDP randomizer is $\delta$-close to an $\eps$-deletion LDP randomizer.
            \begin{lemma}[{\cite[Lemma 3.7]{feldman_hiding_2021}}]
                Fix any $\eps>0,\delta\in[0,1]$. If $R$ is an $(\eps,\delta)$-deletion LDP randomizer with reference distribution $R_0$, then there exists $R'$ such that (i)~$\totalvardist{R(x)}{R'(x)}\leq \delta$ for every $x,x'$, and (ii)~$R'$ is an $\eps$-deletion LDP randomizer with reference distribution $R_0$. In particular, $R'$ is a $2\eps$-replacement LDP randomizer by~\cref{lemma:ldp:deletion:replacement}.
            \end{lemma}

        \subsection{Replacement LDPs} \label{replace-ldps} 

            Recall that symmetric protocol means all LDP randomizers are the same (not their randomness). The following lemma shows how to convert an asymmetric protocol (where each user might use a different local randomizer) into a symmetric one (where all users use the same local randomizer).
            
            \begin{lemma}[Asymmetric to Symmetric LDP. {\cite[Lemma 4]{acharya_test_2018}}]
                Suppose there exists a private-coin (resp., public-coin) mechanism composed of LDP randomizers for some task $\mathcal{T}$ with $n$ users and probability of success $5/6$. Then, there exists a private-coin (resp., public-coin) mechanism with symmetric LDP randomizers for $\mathcal{T}$ with $n'=\bigO{n\log{n}}$ users and probability of success $2/3$. 
            \end{lemma}
            \begin{proof}
                We follow the proof from \cite{acharya_test_2018}. Let $(R_i)_{i \in [n]}$ be the mechanism with $R_i\colon \mathcal{X} \to \mathcal{Y}$ being the local randomizer of the $i$-th user. We create a \emph{symmetric} (randomized) mechanism $R\colon \mathcal{X} \to [n] \times \mathcal{Y}$ defined as follows: On input $x \in \mathcal{X}$: (i) use private randomness to generate $j \in [n]$ uniformly at random, then (ii) output $(j, R_j (x))$.
                
                (If public randomness is available, then $R$ is defined as $R\colon \mathcal{X} \to \mathcal{Y}$ , where $j$ is chosen uniformly at random using the public randomness, and the output is simply $R_j(x)$~--~since all parties have access to the public randomness, there is no need to send $j$ as well.)

                We simulates $W$ if we have each $j = 1..n$ once and report them. Further, by a standard \emph{coupon-collector argument}, requiring $n' = \bigO{n\log{n}}$ we have that with probability $5/6$, each $j \in [n]$ will be drawn at least once. 

                Overall, the probability of failure is at most $1/6 + 1/6 = 1/3$ by union bound (where the first $1/6$ is the failure rate of original mechanism, and the other $1/6$ is the failure rate of the coupon-collector). 
            \end{proof}

            \begin{theorem}[Advanced Grouposition for Pure LDP. {\cite[Theorem 4.2]{bun_heavy_2017}}]
                Let $X \in \mathcal{X}^n, X' \in \mathcal{X}^n$ differ in at most $k$ entries for some $1 \leq k \leq n$. Let $A = (R_1, \ldots, R_n)\colon \mathcal{X}^n \to \mathcal{Y}$, where each $R_i$ is $\eps$-LDP. Then for every $\delta > 0$ and $\eps' = k\eps^2/2 + \eps\sqrt{2k\ln{1/\delta}}$, we have:
                \[
                    \Pr_{y \leftarrow A(X)}
                    \bigg[
                        \ln{\frac{\Pr[A(X) = y]}{\Pr[A(X') = y]}} > \eps' 
                    \bigg]
                    \leq \delta.
                \]
                In particular, for every $\delta > 0$ and every set $T \subseteq \mathcal{Y}$, we have $\Pr[A(X) \in T] \leq e^{\eps'}\Pr[A(X') \in T] + \delta.$
            \end{theorem}
            \begin{proof}
                Without loss of generality, we assume that $X, X'$ differ in the first $k$ entries. We begin with the privacy loss random variable of $A(X), A(X')$:
                \[
                    L_{A(X), A(X')}
                    = \ln{\frac{\Pr[A(X) = y]}{\Pr[A(X') = y]}}
                    = \sum^{k}_{i=1} \ln{ \frac{\Pr[R_i(x_i) = y_i]}{\Pr[R_i(x'_i) = y_i]}}
                    = \sum^k_{i=1} L_{R_i(x_i), R_i(x'_i)}
                \]
                 By taking expectation of $L_{R_i(x_i), R_i(x'_i)}$, it becomes KL-divergence. Which now we can apply the Proposition 3.3 of \cite{buns_concentrated_16} (taking $\alpha = 1$) and acquire:
                \[
                    \expect{L_{A(X), A(X')}}
                    = \sum^k_{i=1} \expect{L_{R_i(x_i), R_i(x'_i)}}
                    \leq \frac{k}{2}\eps^2.
                \]
                Since $R_i$ are $\eps$-LDP randomizers, we have $L_{R_i(x_i), R_i(x'_i)} \in [-\eps, \eps]$. Then by Hoeffding's inequality, for every $t > 0$:
                \begin{align*}
                    \exp{\big(\frac{-t^2}{2k\eps^2}\big)}
                    &\geq \Pr[\sum^k_{i=1} L_{R_i(x_i), R_i(x'_i)} > \sum^k_{i=1} \expect{L_{R_i(x_i), R_i(x'_i)}} + t ]\\
                    &= \Pr[ L_{A(X), A(X')} > \sum^k_{i=1} \expect{L_{R_i(x_i), R_i(x'_i)}} + t ]\\
                    &\geq \Pr[ L_{A(X), A(X')} > k\eps^2/2 + t ] \quad \text{(Since $\expect{L_{R_i(x_i), R_i(x'_i)}}
                    \leq \frac{1}{2}\eps^2$)}.
                \end{align*}
                Thus, we can choose $t=\eps\sqrt{2k\ln{(1/\delta)}}$ such that $\delta = \exp{(-t^2/2k\eps^2)}$, which completes the proof.
            \end{proof}

            Advanced grouposition is a useful in the sense that it presents the implication from LDP to DP. We can see that as an analogue of \cite[Advanced Composition Theorem]{DworkMNS06} as they share similar proof techniques. 

            \begin{theorem}[Advanced Grouposition for Approximate LDP. {\cite[Theorem 4.3]{bun_heavy_2017}}]
                Let $X \in \mathcal{X}^n, X' \in \mathcal{X}^n$ differ in at most $k$ entries for some $1 \leq k \leq n$. Let $A = (R_1, \ldots, R_n): \mathcal{X}^n \to \mathcal{Y}$ be $(\eps, \delta)$-LDP. Then for every $\delta' > 0$ and $\eps' = k\eps^2/2 + \eps\sqrt{2k\ln{1/\delta'}}$, and every set $T \subseteq \mathcal{Y}$, we have: $\Pr[A(X) \in T] \leq e^{\eps'}\Pr[A(X') \in T] + \delta + k \delta'.$
            \end{theorem}

            The following is a composition of LDP randomizers.
            
            \begin{lemma}[Composition of $(\eps, \delta)$-LDP Randomizers. {\cite[Lemma A.1]{erlingsson_encode_2020}}] \label{ldp:composition}
                Assume that for every $(\eps_1, \delta_1)$-LDP randomizer $\mathcal{Q}_1: \{0, 1\} \to \{0, 1\}$ and every $(\eps_2, \delta_2)$-LDP randomizer $\mathcal{Q}_2: \{0, 1\} \to \{0, 1\}$ we have that $\mathcal{Q}_2 \circ \mathcal{Q}_1$ is a $(\eps, \delta)$-LDP randomizer. Then for every $(\eps_1, \delta_1)$-LDP randomizer $R_1: \mathcal{X} \to \mathcal{Y}$ and $(\eps_2, \delta_2)$-LDP randomizer $R_2: \mathcal{Y} \to \mathcal{Z}$, we have that $R_2 \circ R_1$ is an $(\eps, \delta)$-LDP randomizer.
            \end{lemma}
            \begin{proof}
                Let $R_1: \mathcal{X} \to \mathcal{Y}$ be a $(\eps_1, \delta_1)$-LDP randomizer and $R_2: \mathcal{Y} \to \mathcal{Z}$ be a $(\eps_2, \delta_2)$-LDP randomizer. 
                
                We will prove this by contradiction. First, assume that for some $(\eps, \delta)$, there exists an event $S \subseteq Z$ such that for some $x, x'$:
                
                \begin{equation}
                    \Pr[R_2 (R_1 (x)) \in S] > e^\eps \Pr[R_2 (R_1(x')) \in S] + \delta.
                \end{equation}
    
                Then, we will show that there exists an $(\eps_1, \delta_1)$-LDP local randomizer $\mathcal{Q}_1: \{0, 1\} \to \{0, 1\}$ and $(\eps_2, \delta_2)$-LDP local randomizer $\mathcal{Q}_2: \{0, 1\} \to \{0, 1\}$ such that:
    
                \[
                    \Pr[\mathcal{Q}_2 (\mathcal{Q}_1 (x)) = 1] - e^\eps \Pr[\mathcal{Q}_2 (\mathcal{Q}_1 (x')) = 1] > \delta,
                \]
    
                which violates the assumption (that $\mathcal{Q}_2 \circ \mathcal{Q}_1$ is $(\eps, \delta)$-LDP) of the lemma.
    
                Let:
                \[
                    y_0 := \argmin_{y \in \mathcal{Y}} \Pr [R_2(y) \in S],
                \]
                \[
                    y_1 := \argmax_{y \in \mathcal{Y}} \Pr [R_2(y) \in S],
                \]
                \[
                    P_1 := \{ y \in \mathcal{Y} \mid \Pr [R_1(x) = y] - e^\eps \Pr [R_1(x') = y] > 0 \}.
                \]
    
                By using $y_0, y_1$ and our assumption, we can construct the following:
                \begin{align*}
                    & \mleft( \Pr [R_1(x) \notin P_1] - e^\eps \Pr [R_1(x') \notin P_1] \mright) \cdot \Pr[R_2(y_0) \in S]\\
                    & + \mleft( \Pr [R_1(x) \in P_1] - e^\eps \Pr [R_1(x') \in P_1] \mright) \cdot \Pr[R_2(y_1) \in S]\\
                    & \geq \sum_{y \in \mathcal{Y}} \bigg( \Pr[R_2(y) \in S] \cdot \big( \Pr[R_1(x) = y] - e^\eps \Pr[R_1(x') = y] \big) \bigg)\\
                    & = \sum_{y \in \mathcal{Y}} \Pr [R_2(y) \in S \land R_1(x) = y] - e^\eps \sum_{y \in \mathcal{Y}}\Pr[R_2(y) \in S \land R_1(x') = y]\\
                    & = \Pr[R_2 (R_1 (x)) \in S] - e^\eps \Pr[R_2 (R_1(x')) \in S]\\
                    & > \delta.
                \end{align*}
    
                Now we can define $\mathcal{Q}_1(0) := \indic{R_1(x) \in P_1}$ and $\mathcal{Q}_1(1) := \indic{R_1(x') \in P_1}$, where $\indic{\cdot}$ is the indicator function. Since we defined $\mathcal{Q}_1$ by restraining the the set of inputs and post-processing the output, thus $\mathcal{Q}_1$ is a $(\eps_1, \delta_1)$-LDP randomizer as well.
    
                Next, we can define $b \in \{0, 1\}, \mathcal{Q}_2(b) := \indic{R_2(y_b) \in S}$. It is also easy to see that $\mathcal{Q}_2$ is a $(\eps_2, \delta_2)$-LDP randomizer. By using the previous result, we have:
                \begin{align*}
                    & \Pr[\mathcal{Q}_2 (\mathcal{Q}_1 (x)) = 1] - e^\eps \Pr[\mathcal{Q}_2 (\mathcal{Q}_1 (x')) = 1]\\
                    & = \sum_{b \in \{0, 1\}} \bigg( \Pr[\mathcal{Q}_2(b) = 1] \cdot \big( \Pr[\mathcal{Q}_1(0) = b] - e^\eps \Pr[\mathcal{Q}_1(1)] = b \big) \bigg)\\
                    & = \big( \Pr[R_1(x) \notin P_1] - e^\eps \Pr[R_1(x') \notin P_1] \big)\cdot \Pr[R_2(y_0) \in S]\\
                    & + \big( \Pr[R_1(x) \in P_1] - e^\eps \Pr[R_1(x') \in P_1] \big)\cdot \Pr[R_2(y_1) \in S]\\
                    & > \delta.
                \end{align*}
                which is what we want to show for contradiction.
            \end{proof}

            Following corollary identifies the value of $\eps$ for $R_2 \circ R_1$.
    
            \begin{corollary}[Followed from \cref{ldp:composition}. {\cite[Corollary A.2]{erlingsson_encode_2020}}] \label{composition-ldp-weaker-result}
                For every $\eps_1$-LDP randomizer $R_1: \mathcal{X} \to \mathcal{Y}$ and every $\eps_2$-LDP randomizer $R_1: \mathcal{Y} \to \mathcal{Z}$, we have that $R_2 \circ R_1$ is a $\eps$-LDP randomizer, where $\eps \eqdef \ln{\big(\frac{e^{\eps_1 + \eps_2}+1}{e^{\eps_1} + e^{\eps_2}}\big)}$. In addition, if $R_1$ is removal $\eps_1$-LDP, then $R_2 \circ R_1$ is a removal $\eps$-LDP.
            \end{corollary}
            \begin{proof}
                We can reduce the problem to $\mathcal{X} = \mathcal{Y} = \mathcal{Z} = \{0, 1\}$ using \cref{ldp:composition}. We choose to bound the case of $R_2 \circ R_1 (b) = 1$ for $b \in \{0, 1\}$:
                \[
                    \frac{\Pr[R_2 \circ R_1 (0) = 1]}{\Pr[R_2 \circ R_1 (1) = 1]} = \frac{\Pr[R_1 (0) = 0] \cdot \Pr[R_2 (0) = 1] + \Pr[R_1 (0) = 1] \cdot \Pr[R_2 (1) = 1]}{\Pr[R_1 (1) = 0] \cdot \Pr[R_2 (0) = 1] + \Pr[R_1 (1) = 1] \cdot \Pr[R_2 (1) = 1]}.
                \]
                Defining $p_0 \eqdef \Pr[R_1(0) = 0], p_1 \eqdef \Pr[R_1(1) = 0], \alpha = \Pr[R_2(0) = 1] / \Pr[R_2(1) = 1]$, the above expression simplifies to:
                \[
                    \frac{1 + (\alpha - 1)p_0}{1 + (\alpha - 1)p_1}.
                \]
                Since we want to find $\eps$ such that:
                \[
                    \frac{\Pr[R_2 \circ R_1 (0) = 1]}{\Pr[R_2 \circ R_1 (1) = 1]} \leq e^\eps,
                \]
                we came to solve the following problem:
                \begin{align*}
                    \max_{p_0, p_1} \quad
                        & \frac{1 + (\alpha - 1)p_0}{1 + (\alpha - 1)p_1}\\
                    \textrm{s.t.} \quad
                        & \alpha = e^{\eps_2}\\
                        & p_0 > p_1\\
                        & e^{-\eps_1} \leq \frac{p_0}{p_1} \leq e^{\eps_1}\\
                        & e^{-\eps_1} \leq \frac{1-p_0}{1-p_1} \leq e^{\eps_2}.
                \end{align*}
                Note that we can assume without loss of generality that $\alpha \geq 1$ and thus the expression is maximized while $\alpha = e^{\eps_2}$ and $p_0 > p_1$. Solving this optimization problem yields that $\eps \eqdef \ln{\big(\frac{e^{\eps_1 + \eps_2}+1}{e^{\eps_1} + e^{\eps_2}}\big)}$.

                The deletion LDP case is achieved if we substitute $R_1(x')$ with the reference distribution $R_0$.
            \end{proof}

            \cite[Theorem 4.3]{naor_can_2020} proves the same result of $\eps$ and shows that the bound is tight. They also gives a weaker result of $\eps \eqdef \frac{1}{2}\eps_1 \eps_2$ which is easier to manipulate then $\eps \eqdef \ln\!{\big(\frac{e^{\eps_1 + \eps_2}+1}{e^{\eps_1} + e^{\eps_2}}\big)}$. See proof of $\ln\!{\big(\frac{e^{\eps_1 + \eps_2}+1}{e^{\eps_1} + e^{\eps_2}}\big)} \leq \frac{1}{2}\eps_1\eps_2$ for every $\eps_1,\eps_2>0$ in Appendix, \cref{composition-ldp-weaker-result-proof}.

            \begin{theorem}[Approximate LDP to pure LDP {\cite[Theorem~6.1]{bun_heavy_2017}}]
                Let $\eps\in(0,1/4]$ and $\delta\in[0,1]$, and suppose $R_1,\dots, R_n$ are the $(\eps,\delta)$-LDP randomizers. Then there exists a \emph{public-coin} protocol with randomizers $R'_1,\dots, R'_n$ such that, for every choice of integer
                \[
                    5\ln\frac{1}{\eps} \leq T \leq \frac{1-e^{-\eps}}{4\delta n e^\eps} 
                \]
                \begin{enumerate}
                    \item $R'_i$ is $10\eps$-LDP for every $i\in[n]$
                    \item for every dataset $x = (x_1,\dots,x_n)\in\mathcal{X}^n$, the total variation distance between the distribution $D_x$ of messages $(R_1(x_1),\dots, R_n(x_n))$ and the distribution $D'_x$ of $(R'_1(x_1),\dots, R'_n(x_n))$ 
                    satisfies
                    \[
                        \totalvardist{D_x}{D'_x} \leq n \mleft( \mleft(\frac{1}{2}+\eps\mright)^T + 6T\delta\cdot \frac{e^\eps}{1-e^{-\eps}}  \mright)
                    \]
                    \item Each user sends at most $\log_2 T$ bits of communication.
                    \item The protocol uses $T\sum_{i=1}^n r_i$ bits of public randomness, where $r_i$ is the number of random bits required by $R_i$.
                \end{enumerate}
            \end{theorem}

    \section{Relation between Shuffle DP and LDP} \label{shuffle_to_ldp}

        \begin{definition}[One-Message Shuffle Model]
            In the one-message shuffled model, each user sends $m=1$ message.
        \end{definition}

        \begin{theorem}[From One-Message Shuffle Model to LDP. \cite{cheu_distributed_2019}, Theorem 6.2]
        \label{theo:shuffle_to_ldp}
            If $P_n = (R, S, A)$ is a one-message shuffled model with $n \in \N$ users that satisfies $(\eps_S, \delta_S)$-DP, then the local randomizer $R$ satisfies $(\eps_L, \delta_L)$-LDP for $\eps_L = \eps_S + \ln{n}$ and $\delta_L = \delta_S$. Therefore, the symmetric local protocol $P_L = (R, A \circ S)$ satisfies $(\eps_L, \delta_L)$-DP.
        \end{theorem}
        \begin{proof}
            Let $\mathcal{E}_{R,\delta_S}$ denote the set of parameters $\eps'>0$ for which $R\colon \mathcal{X} \to \mathcal{Y}$ is \emph{not} $(\eps', \delta_L)$-LDP, and let $\eps$ denote the supremum of $\mathcal{E}_{R,\delta_S}$. We can assume without loss of generality that $\eps\in \mathcal{E}_{R,\delta_S}$ (otherwise, we can use the argument below with $\eps' \eqdef \eps(1-\alpha)$ for some small $\alpha$ and take the limit as $\alpha\downarrow 0$ in the end).
            \begin{equation}
                \label{eq:implication:notldp}
                \Pr[R(x') \in Y] > e^\eps \Pr[R(x) \in Y] + \delta_L
            \end{equation}
            for some $x, x' \in \mathcal{X}$ and $Y \subseteq \mathcal{Y}$. We define $p' \eqdef \Pr[R(x') \in Y]$ and $p \eqdef \Pr[R(x) \in Y]$ for brevity, so that the above becomes $p' > e^\eps p + \delta_S$.

            Then we define the set $\mathcal{W} \eqdef \{W \in \mathcal{Y}^n \mid \exists i, w_i \in Y\}$, which is a set of output for $P$ where any of its randomizer is not local differentially private. Construct two databases of size $n$: $X \eqdef (x,  x,\ldots, x)$ and $X' \eqdef (x', x, \ldots, x)$.

            Since $P_n$ is $(\eps_S, \delta_S)$-differentially private, we can write:
            \begin{equation}
                \label{eq:implication:dp}
                \Pr[P_n (X') \in \mathcal{W}] \leq e^\eps_S \Pr[P_n (X) \in \mathcal{W}] + \delta_S
            \end{equation}

            Now we have:
            \begin{align}
                \Pr[P_n (X) \in \mathcal{W}]
                &= \Pr[S(R(x), \ldots, R(x)) \in \mathcal{W}] \notag\\
                &= \Pr[(R(x), \ldots, R(x)) \in \mathcal{W}] \notag\\
                &= \Pr[\exists i, R(x) \in Y] \notag\\
                &\leq n\Pr[R(x) \in Y] = np, \label{eq:shuffle_to_ldp:publiccoinstep}
            \end{align}
            where the second equality is because $\mathcal{W}$ is closed under permutation, and the inequality results from the union bound.

            Similarly, we have:
            \begin{align*}
                \Pr[P_n (X') \in \mathcal{W}]
                &= \Pr[(R(x'), R(x), \ldots, R(x)) \in \mathcal{W}]\\
                &\geq \Pr[R(x') \in Y] = p'\\
                &> e^\eps \cdot p + \delta_S
            \end{align*}
            where the last inequality comes from~\eqref{eq:implication:notldp}. Finally we can rewrite~\eqref{eq:implication:dp}:
            \begin{align*}
                e^\eps p + \delta_S < \Pr[P_n (X') \in \mathcal{W}]
                &\leq e^{\eps_S} \Pr[P_n (X) \in \mathcal{W}] + \delta_S\\
                &\leq e^{\eps_S} np + \delta_S,
            \end{align*} Finally we obtain:
            \[
                \eps < \eps_S + \ln{n}.
            \]
            This means that the largest possible $\eps$ for which $R$ is not $(\eps,\delta_S)$-DP is strictly smaller than $\eps_L \eqdef \eps_S + \ln{n}$. So $R$ is $(\eps_L,\delta_S)$-DP.
        \end{proof}
        
        Note that the above proof applies to both public-coin and private-coin protocols, as the only key step where the possibly joint randomization is taken into account is~\eqref{eq:shuffle_to_ldp:publiccoinstep}. It is not clear whether the analysis can be significantly improved for private-coin protocols, for which $np$ can be replaced by $1-(1-p)^n$ in the RHS of~\eqref{eq:shuffle_to_ldp:publiccoinstep}.

%

        \begin{theorem}[From LDP to shuffler. {\cite[Theorem 3.2]{feldman_stronger_2022}}]
            Let $R$ be an $\eps_L$-differentially private local randomizer. Let $S = (R, \pi)$ be the shuffler that given a dataset, samples a uniform random permutation $\pi$ over $[n]$, then sequentially compute reports $z_i = R^{(i)}(z_{i-1}, x_{\pi (i)})$ for $i \in [n]$ and outputs $z_{1:n}$. Then for any $\delta \in [0, 1]$ such that $\eps_L \leq \ln{(\frac{n}{8\ln(2/\delta)} - 1)}$. $S$ is $(\eps, \delta)$-DP, where:
            \begin{equation}
                \eps \leq \ln\!
                \bigg(
                    1 + 4(e^{\eps_L} - 1)
                    \bigg( \sqrt{\frac{2\ln (4 / \delta)}{(e^{\eps_L} + 1)n}} + \frac{1}{n} \bigg)
                \bigg)
            \end{equation}
            In particular, for $\eps_L \leq 1$, this gives $\eps = \bigO{\eps_L \sqrt{ \frac{\log(1/\delta)}{n} }}$, and for $\eps_L \geq 1$ we get $\eps = \bigO{ e^{\eps_L/2}\sqrt{ \frac{\log(1/\delta)}{n} }}$.
        \end{theorem}
        We redirect readers to their work for the proof, and the extension when the local randomizer is itself approximate LDP.
        
        On a side note, above theorem holds true for robust shuffle privacy if we replace $n = \gamma  n$ for $\gamma \in (0, 1]$ being the fraction of users that actually follows the protocol (i.e., non-malicious). Indeed, this follows by applying exactly the same argument (amplification b shuffling) as in the proof of the theorem, but only to the $n' \eqdef \gamma n$ honest users.
        
        Finally, it is worth noting that the original LDP guarantee \emph{does} of course still hold after shuffling; and so in particular the output of the shuffler is both $(\eps, \delta)$-DP and $\eps_L$-DP.

    \section{Privacy Amplification by Subsampling}

        The chapter ``Composition of Differential Privacy \& Privacy Amplification by Subsampling'' \cite[Section 6]{steinke_composition_subsampling} written by Thomas Steinke contains more in-depth proofs and claims on this topic. Here, we state a few that is important and provide proofs in alternative methods (which is somewhat simpler and weaker).
    
        \begin{theorem}[Privacy Amplification by Subsampling for Approximate DP. {\cite[Theorem 28]{steinke_composition_subsampling}}]
            Let $U \subseteq [n]$ be a random subset and $X$ be a dataset with $n$ records. Define $M\colon \mathcal{X}^n \to \mathcal{Y}$ to satisfy $(\eps, \delta)$-DP. Construct a mechanism $M^U\colon \mathcal{X}^n \to \mathcal{Y}$ where it first samples $|U|$ records from $X$ to form the smaller dataset $X_U$ and then compute $M(X^U)$. We further let $p = \max_{i \in [n]} \Pr_U[i \in U]$. Then $M^U$ is $(\eps', \delta')$-DP for $\eps' = \ln(1 + p(e^\eps - 1))$ and $\delta' = p\delta$.
        \end{theorem}
        
        A more general version of this theorem is defined in \cite[Theorem 28]{steinke_composition_subsampling}. The complete proof is available in \cite{steinke_composition_subsampling}, but here we present a self-contained version for the case of $p = m/n$, using the same notation as lecture notes of Jonathan Ullman \footnote{From CS7880 Homework 1 Solution: \url{http://www.ccs.neu.edu/home/jullman/cs7880s17/HW1sol.pdf} (Accessed Nov 21, 2022).}.

        \begin{proof}
            For datasets $X \sim X'$, let $p = \frac{m}{n} = \Pr[i \in U]$. The goal is to show that $M^U$ is $(\ln{(1 + p(e^\eps - 1))}, p\delta)$-DP. Fix any measurable $S\subseteq \mathcal{Y}$: to show the result, we need to upper bound the following ratio by $1 + p(e^\eps - 1))$:
            \begin{align*}
                \frac{\Pr[M_U(X) \in S] - \delta p}{\Pr[M_U(X') \in S]}
                &= \frac{\Pr[M(X_U) \in S] - \delta p}{\Pr[M({X_U'}) \in S]}\\
                &= \frac{\Pr[M(X_U) \in S \mid i \in U]p + \Pr[M(X_U) \in S \mid i \notin U](1-p) - \delta p}{\Pr[M(X_U') \in S \mid i \in U]p + \Pr[M(X_U') \in S \mid i \notin U](1-p)}
            \end{align*}
            For convenience, we set:
            \begin{align*}
                C &= \Pr[M(X_U) \in S \mid i \in U]\\
                C' &= \Pr[M(X_U') \in S \mid i \in U]\\
                D &= \Pr[M(X_U) \in S \mid i \notin U] = \Pr[M(X_U') \in S \mid i \notin U]
            \end{align*}
            Now we have:
            \[
                \frac{\Pr[M_U(X) \in S] - \delta p}{\Pr[M_U(X') \in S]} = \frac{pC + (1-p)D - p\delta}{pC' + (1-p)D}
            \]
            \begin{align*}
                \Pr[M_U(X) \in S] - \delta p
                &= pC + (1-p)D - p\delta\\
                &\leq p \big( e^\eps \min(C', D) + \delta \big) + (1-p)D - p\delta\\
                &= p \big( \min(C', D) + (e^\eps - 1)\min(C', D) + \delta \big) + (1-p)D - p\delta\\
                &\leq p \big( \min(C', D) + (e^\eps - 1)(pC' + (1-p)D) + \delta \big) + (1-p)D - p\delta\\
                &\leq p \big( C' + (e^\eps - 1)(pC' + (1-p)D) + \delta \big) + (1-p)D - p\delta\\
                &= p \big( C' + (e^\eps - 1)(pC' + (1-p)D) \big) + (1-p)D\\
                &= \big( pC' + (1-p)D \big) \times \big( p(e^\eps - 1) + 1 \big)\\
                &= \exp{(\ln{(p(e^\eps - 1) + 1)})}\big( pC' + (1-p)D \big)
            \end{align*}
            The first inequality follows from the definition of $(\eps, \delta)$-DP. Note that $C \leq e^\eps D + \delta$ because $X_U$ with $x_i$ is still neighboring to subset of $X$ of size $U$ without $x_i$. The second inequality uses the fact that $\min(x, y) \leq \alpha x + (1-\alpha)y, \forall \alpha \in [0, 1]$, where we use $\alpha = p$. This simpler proof get the same bound as \cite{steinke_composition_subsampling} by applying $e^{\ln}$ on the second-to-last line.

            From above, we finally get:
            \begin{align*}
                \frac{\Pr[M_U(X) \in S] - \delta p}{\Pr[M_U(X') \in S]}
                &\leq \exp(\ln(p(e^\eps - 1) + 1)),
            \end{align*}
            which concludes the proof.
        \end{proof}

    \printbibliography

    \section{Appendix}
        \subsection{Proof of the weaker result in \cref{composition-ldp-weaker-result}} \label{composition-ldp-weaker-result-proof} 
            Recall that the goal is to prove:
            \[
                \frac{1+e^{x+y}}{e^x + e^y} \leq \exp{\mleft(\frac{xy}{2}\mright)}.
            \]
            \begin{proof}\footnote{This proof is due to Thomas Steinke. See \href{https://twitter.com/ccanonne_/status/1564500476200382464?s=20&t=pVsmb0QxtQXn6Pf3OntG5A}{here} for the proof in emoji flavor.}
                Suppose w.l.o.g that $0 \leq x \leq y$. Consider the random variable $Z$ such that:
                \[
                    \Pr[Z = x] = \frac{e^y}{e^x + e^y} \quad \text{and} \quad \Pr[Z = -x] = \frac{e^x}{e^x + e^y}.
                \]
                Then $|Z| \leq x$ and $\expect{Z} = x \cdot \frac{e^y - e^x}{e^x + e^y}$. We also have, on the one hand, that
                $\expect{e^Z} = \frac{1 + e^{x+y}}{e^x + e^y}$.
                
                On the other hand, by Hoeffding's Lemma 
                    \footnote{Hoeffding's Lemma: Let $X$ be any real-valued random variable such that $X \in [a, b]$ (w. prob 1). Then, for all $\lambda \in \R$, $\expect{e^{\lambda X}} \leq \exp{\mleft( \lambda \expect{X} + \frac{\lambda^2 (b - a)^2}{8} \mright)}$}
                with $\lambda = 1$:
                \begin{align*}
                    \expect{e^Z}
                    &\leq \exp\mleft( \expect{Z} + \frac{(2x)^2}{8}\mright)\\
                    &= \exp\mleft( x \cdot \frac{e^y - e^x}{e^x + e^y} + \frac{x^2}{2}  \mright)\\
                    &= \exp\mleft( x \cdot \frac{e^{y - x} - 1}{e^{y - x} + 1} + \frac{x^2}{2}  \mright).
                \end{align*}

                Then we apply this inequality: for every $t \geq 0$, $\frac{e^t - 1}{e^t + 1} \leq \frac{t}{2}$. \footnote{This can be proven by the concavity of $f\colon t \to \frac{e^t-1}{e^t + 1}$ on $[0, \infty)$, so $f(t) \leq f(0) + f'(0)t = \frac{t}{2}$.} With $t:= y-x$ and substituting the above result, we get
                \begin{align*}
                    \frac{1 + e^{x+y}}{e^x + e^y} 
                    &\leq \exp\mleft( x \cdot \frac{e^{y - x} - 1}{e^{y - x} + 1} + \frac{x^2}{2}  \mright)\\
                    &\leq \exp\mleft( x \cdot \frac{(y-x)}{2} + \frac{x^2}{2}  \mright)\\
                    &= \exp\mleft( \frac{xy}{2} \mright). \qedhere
                \end{align*}                
            \end{proof}

\end{document}